%% file: paper.tex
\documentclass{sig-alternate-05-2015}
\pdfoutput=1 
\usepackage{amsmath}
\usepackage{microtype}
\usepackage{amsfonts}
\usepackage{enumitem}
\usepackage{xspace}
\usepackage{placeins}
\usepackage{units}
\usepackage{bm}
\usepackage{enumitem}

\title{The Limits of Popularity-Based Recommendations,\\ and the Role of Social Ties\titlenote{\textcopyright ACM, 2016. This is the author's version of the work. It is posted here by permission of ACM for your personal use. Not for redistribution. The definitive version will appear in the proceedings of KDD 2016. http://doi.acm.org/10.1145/2939672.2939797}}
\date{\today}

\newtheorem{lemma}{Lemma}
\newtheorem{theorem}{Theorem}
\newtheorem{corollary}{Corollary}
\newtheorem{definition}{Definition}
\newtheorem{fact}{Fact}

\newcommand{\xlim}{\bm{x}^{\infty}}

\newcommand{\E}{\mathbf{E}}
\newcommand{\0}{\bm{0}}
\newcommand{\1}{\bm{1}}
\newcommand{\I}{\bm{I}}
\newcommand{\M}{\bm{M}}
\newcommand{\A}{\bm{A}}

\newcommand{\boldf}{\bm{f}}
\newcommand{\Mthin}{M}

\newcommand{\bH}{\mathbf{H}}
\newcommand{\bL}{\bm{L}}
\newcommand{\bh}{\mathbf{h}}
\newcommand{\bb}{\bm{b}}

\newcommand{\x}{\mathbf{x}}

\newcommand{\bgamma}{\bm{\gamma}}

\author{
Marco Bressan\thanks{\small Supported in part by a Google Focused Research Award, by the Sapienza Grant C26M15ALKP, by the SIR Grant RBSI14Q743, and by the ERC Starting Grant DMAP 680153.}\;\thanks{Corresponding author. Address: {\tt bressan@di.uniroma1.it}}\\
Sapienza University of Rome\\
Rome, Italy\\
\and
Stefano Leucci\footnotemark[2]\\ 
Sapienza University of Rome\\
Rome, Italy\\
\and
Alessandro Panconesi\footnotemark[2]\\ 
Sapienza University of Rome\\
Rome, Italy\\
\and
Prabhakar Raghavan\\
Google\\
Mountain View, CA
\and
Erisa Terolli\footnotemark[2]\\
Sapienza University of Rome \\
Rome, Italy \\
}

\begin{document}


\CopyrightYear{2016}
\conferenceinfo{KDD '16,}{August 13 - 17, 2016, San Francisco, CA, USA}
\isbn{978-1-4503-4232-2/16/08}\acmPrice{\$15.00}
\doi{http://dx.doi.org/10.1145/2939672.2939797}

\maketitle

\begin{abstract}
In this paper we introduce a mathematical model that captures some of the salient features of recommender systems that are based on popularity and that try to exploit social ties among the users.  We show that, under very general conditions, the market always converges to a steady state, for which we are able to give an explicit form. Thanks to this we can tell rather precisely how much a market is altered by a recommendation system, and determine the power of users to influence others. Our theoretical results are complemented by experiments with real world social networks showing that social graphs prevent large market distortions in spite of the presence of highly influential users.
\end{abstract}


\input{introduction.tex}

\input{relatedworks.tex}
\input{theory.tex}

\input{experiments.tex}

\vspace{20pt}
\noindent
{\bf Acknowledgments.}
The authors are grateful to Tiziana Castrignan\`o at CINECA for her precious help with the technical setup of the experiments, and to Flavio Chieri\-chet\-ti for the stimulating and insightful discussions throughout the development of this work.

\bibliographystyle{abbrv}
\small
\input{formatted.bbl} 

\normalsize
\appendix
\input{appendix.tex}

\end{document}

%% file: introduction.tex
\section{Introduction}
Recommender systems (RS's) are a paradigmatic example of the interaction between humans and algorithms in the cultural arena. 
From YouTube videos to books on Amazon and movies on iTunes or Netflix, online choices are increasingly mediated by such algorithms. RS's are popular for their ability to conjure up non-trivial relationships between products or news stories, and as such they are becoming powerful economic actors.
It is this economic dimension of RS's that we intend to explore in this paper. 

A fundamental question in this context, and the one we intend to tackle in this paper, is the following: to what extent can a market be altered by a RS?
For instance, assume that an online bookstore starts adopting a RS.
Will unknown books become hits and vice-versa?
How and how much will the habit of the typical reader change?
A large body of work has investigated questions of this kind at the individual, user level \cite{Cosley&2003,salganik2006,Sharma&2013,zhu2014}. 
For instance, a nice experiment described in \cite{salganik2006} suggests that even a simple type of feedback, such as providing the ranking based on the number of downloads, may significantly boost market shares.
Studies such as these provide precious insights but it is difficult to derive from them quantitatively accurate predictions about markets in the long-run.
Also, it is entirely conceivable that the amplification effects observed be temporary, or not transferable to markets whose sheer size and complexity dwarfs that of artificial settings. 
Other studies have tried to draw conclusions from behavioural data of real markets \cite{Ehrmann&2008,Duan&2008,Duan&2008b,Nguyen&2014}, but the paucity of available data makes them tentative.

In this paper we try a different approach by introducing a natural model for markets that are governed by a RS. 
Our model is simple enough to allow for a precise mathematical analysis of long-term behaviours, while at the same time it captures some of the relevant features of RS's.
Specifically, there are two ingredients that are known to play an important role in RS's and that we want our model to represent.
The first is popularity, {\em i.e.} the degree of success of an item in its market.
Several studies indicate that popularity feedback ({\em e.g.} number of downloads, user ratings, number of views, etc) can be a powerful determinant of online behaviour \cite{salganik2006,Celma&2008,zhu2014}.
The second ingredient is, for lack of a better terminology, the ``web of kinship'' among users. Social ties are the paradigmatic example here.
They are recognized to be an important factor that shape user choices, and recommenders try to leverage them in a variety of ways \cite{Ma&2009,Ye&2012,Zhao&2013,Chaney&2015}.
More generally RS's try to exploit different types of similarities between users, such as explicit links in an online social network, similarity between patterns of consumer choice, or even the result of complex computations like matrix factorization.
All these diverse situations can be simply and  usefully represented, as we show in this paper, by assuming the existence of a (weighted) network connecting the users.

In simplified form, our mathematical model of market evolution is as follows. In the market we have users (or buyers), connected via a network, and products. When a user buys a product, it either follows its personal inclinations (mode led by a private probability distribution over the products) or it follows a recommendation with a certain probability. When it does, it consults a random neighbour in the network, the \textit{recommender node}. The recommender node picks a random product from its own list of products purchased so far, with a probability that is proportional to the popularity of the product in that moment, and the item so selected is bought by the current buyer.

The main theoretical result in this paper is that this type of system, encompassing a large class of recommenders, always converges to a unique, stationary limit. The result holds under very general conditions. First, users can buy products at different rates.
User $u$, say, can be ten times faster than user $v$ and one thousand times slower than user $w$ etc.
Second, the probability of following the recommendation may change from person to person -- each buyer can have his/her own personal level of trust in recommendations,  and also trust his/her friends to differing degrees. 
Third, the probability with which a product is recommended does not need to be uniform, but it can depend on time.
For instance, we could recommend only the last ten items bought, or skew the probability with which an item is recommended in many other ways. 
Finally, the graph can be any directed graph (so $u$ could follow $v$'s recommendations but not vice versa) and users can have their own personal history of purchases before the recommendation mechanism starts operating. 
None of these affects the result-- the market always converges to a unique limit. 

It is natural to think of the network connecting the users as a social network, but it can in fact be used to model many more scenarios.
For instance, we can make the popularity of an item directly proportional to its number of sales in the whole market by assuming the network to be a clique, or we can represent a powerful super-user that influences the entire network without being affected by it by assuming the network to be an oriented star.

A nice feature of our result is that the limit is given explicitly, in closed form. This makes it possible to explore important properties of the steady state both analytically and computationally. For instance, we can answer easily our motivating question--  can the recommendation mechanism alter the market in a significant way? We can easily tell the distance, in whatever norm, between the initial market shares of every product and those at steady state. We can gain insight into what type of topologies amplify or dampen distorting effects, and we can also carry out a fine grained analysis that tells,  for each user, how much it affects the final outcome. 
A by-product of our analysis is a quantitative notion of user influence that pinpoints exactly how much each user affects the market at the limit. Interestingly,  this notion turns out to be a generalization of Personalized PageRank (PPR) and it is related to the Shapley values of a coalition game that can be defined in our model.

Our limit theorems provide a rich landscape for computational explorations with real-world social networks.  
We explore the distorting effects of various graph topologies and the extent to which a user can influence the market. What emerges is a rather nuanced landscape that is spelled out in full in the rest of the paper. Perhaps the most interesting outcome is that real social networks prove themselves to be a bulwark against market distortion, in spite of the presence of celebrities in their midst.

To summarize, this paper is a step in the direction of a systematic investigation of the long-term effects of recommender systems on the markets in which they are operating. 
This is an important and fascinating topic that deserves more attention that it has gotten so far. We hope that this paper provides some actionable insight into this important problem.

The paper is organized as follows. In \S~\ref{s:rw} we discuss the relevant literature, in \S~\ref{s:t} we define our model and analyse its properties mathematically. We conclude with experiments in \S~\ref{sec:experimental}.

%% file: relatedworks.tex
\section{Related Work}\label{s:rw}
Product popularity is known to have powerful feedback effects on user behaviour and choices, and is leveraged by many existing RS's.
In one of the most eloquent demonstrations,~\cite{salganik2006} shows that, in a digital cultural market, just reporting the ranking based on the number of downloads significantly affects user choices and increases the market's inequality and unpredictability.
Another telling illustration,~\cite{zhu2014}, proves that users are likely to change their mind and reverse their previous choices when given information about the popularity of items. 
A study of friend recommendation in Twitter,~\cite{Su&2016}, shows that recommendations of popular users are much more likely to be accepted than recommendations of ``average'' users.
Several existing online marketplaces make use of popularity-based recommendations -- for instance Amazon\footnote{\small\url{http://www.amazon.com}}, Yelp\footnote{\small\url{http://www.yelp.com}}, TripAdvisor\footnote{\small\url{http://www.tripadvisor.com}}, SoundCloud\footnote{\small\url{http://www.soundcloud.com}} and Last.fm\footnote{\small\url{http://www.last.fm}} all suggest items that are currently trending.

Social ties are also understood to play a key role in the recommendation and purchase processes, and are increasingly exploited by modern RS's.
A study on the Epinions\footnote{\small\url{http://www.epinions.com}} network~\cite{Chua&2011} finds that a large part of the items adopted by a user are also adopted by his/her friends.
Similarly, a study on the Yahoo!~Pulse network~\cite{Yang&2011} shows that a user's interests are highly correlated with, and can be predicted by, the interests of his/her social ties -- and vice versa.
An emblematic example of RS based on social ties is~\cite{Ye&2012}, which uses a generative influence model where users probabilistically pick items according to both their and their friends' preferences.
Another study,~\cite{Chaney&2015}, builds a recommendation algorithm assuming that the two forces driving a user's actions are latent preferences and the influence of friends.
Several other works show or propose that recommendations take place over some social network structure~\cite{Jamali&2009,Jamali&2010,Zhao&2013,Huang&2013,Zhao&2014}.

On the other hand, investigations on the long-term market effects of RS's are scarce; as long as mathematical models are concerned, existing work seems confined to the economic literature.
To the best of our knowledge, the only theoretical investigation on the market effects of RS's is~\cite{Fleder&2009}.
In their work, a single user repeatedly buys one of two products according to either his/her personal taste or to a recommendation consisting in the most-bought product so far; depending on the choice of parameters, the process can lead to a concentration of sales.
There are many differences between their model and ours -- in particular, our model allows users to influence each other.
Two earlier studies,~\cite{Degroot1974} on consensus reaching and~\cite{Friedkin&1990} on opinion formation, investigated the convergence of personal beliefs of users interacting in a social group, similarly to how we investigate convergence in a market with social ties.
Their models are however far from ours in many aspects; for instance, ours is stochastic and allows for complex dependences on past events.
Another paper,~\cite{HervasDrane2015}, studies the effects of word-of-mouth on product sales, but focusing on equilibrium strategies for maximizing user satisfaction.
Finally, a few other works have analysed the interactions between RS's and markets, but either with different aims~\cite{KPR2001,Prawesh&2011} or by simulation instead of analysis~\cite{Zeng&2015}.


%% file: theory.tex

\section{A Model for Recommender Systems}\label{s:t}
Let us begin with the formal definition of our model.

\paragraph{Products} We have a set $\mathcal{P}$ of $m$ products; each product can be bought multiple times during the purchasing process. Products are bought one at a time in an infinite sequence of time steps $t=1,2,\ldots$. The purchasing process is specified below. 
\newline

\paragraph{Users (or Buyers)} We have a set $\mathcal{U} := \{1,\ldots,n\}$ of $n$ users. 
To each user $u$ is associated:
\begin{itemize}[leftmargin=15pt]\itemsep0pt
 \item a fixed purchasing rate $f_u \in (0, 1)$ representing the rate at which $u$ buys products. The $f_u$'s form a probability distribution, i.e.\ $\sum_{u \in \mathcal{U}} f_u = 1$. The precise 
 role of $f_u$ is clarified below when the purchasing process is defined
 \item a probability distribution $B_u$ over the set of products $\mathcal{P}$ reflecting its {\em personal preferences}. The probability $B_u(p), p \in \mathcal{P}$ is $u$'s {\em personal preference} for $p$.
 \item a fixed probability $\alpha_u \in [0, 1)$; when $u$ buys a product it follows a recommendation with probability $\alpha_u$ (see the definition of the purchasing process below)
 \item a list of products purchased in the past (before time $0$), initially containing $k_u > 0$ items. This list can be arbitrary
\end{itemize}

\paragraph{The graph}
Users are connected via a directed, weighted graph $G = (\mathcal{U}, E)$, possibly with self-loops.
Intuitively, an edge $vu$ denotes the fact that $v$ can influence $u$ when $u$ decides to follow a recommendation, or alternatively, that  $u$ ``trusts'' $v$.
The arc has weight $w_{vu}$ which gives the strength of this trust relationship.
The weights are normalised: the weights of the arcs entering each vertex $u$ (if any) sum up to one. 

It is natural to think of this graph as a social network among the users, but it is in fact a very versatile tool with which we can model several different situations.
For instance, when the graph is a clique and the weights are uniform we have the traditional popularity feedback -- an item is recommended with a probability that is proportional to the number of sales in the whole market.
Or, the graph could embed similarities between users, with the weights reflecting the degree of similarity, and so on so forth.
The graph is a flexible tool thanks to which several situation can be cast in our framework and network effects be explored.

\paragraph{The purchasing process}
At each time step $t = 1, 2, \ldots$, a user $u$ is chosen with probability $f_u$ and buys a product in one of two ways:
\begin{itemize}[leftmargin=15pt]\itemsep0pt
\item
With probability $1-\alpha_u$, it chooses a product at random according to its own distribution $B_u$ of personal preferences. 
\item
Alternatively, with probability $\alpha_u$, it follows the recommendation:
\begin{itemize}[leftmargin=10pt]
\item
First, $u$ picks a {\em recommender node} $v$ at random among those pointing to him (if $u$ has no entering arcs then $\alpha_u=0$). The recommender node $v$ is picked with probability $w_{vu}$, the weight of the arc $vu$. 
\item
Once the recommender node $v$ is chosen, a product $p$ is selected at random from the multiset of products acquired by $v$ up to that moment. Let us denote for now this probability simply as $x^t_v(p)$ and postpone to later its precise definition.
\item
The product $p$ so selected is bought, i.e.\ it is added to the multiset of products purchased by the current buyer $u$.  
\end{itemize}
\end{itemize}
Note that even if we say for convenience that the $i$-th purchase takes place at ``time'' $i$,  we make
no assumptions about time itself (i.e. the wall clock time between any one purchase and the next can be arbitrary). The process simply stipulates that items are bought one after the other, in an infinite sequence.

\paragraph{The weight of history}
To conclude the description of the model we need to define the probability distribution, that we denoted as $\{x^t_v(p)\}_{p \in \cal P}$, with which the products of the recommender node are picked.
A reasonable choice to model popularity-based recommendations would be simply to select a product uniformly at random from $v$'s multiset of past purchases.
In this way popular items have a higher probability of being recommended.
But we want to take a more general approach by allowing time-dependent weights.
For instance, we would like to be able to deal with a situation in which, say, only the ten most recent purchases matter, or where the influence of purchases fades away with time in a specific way, and so on so forth.
To this end, focus on a vertex $u$ and recall that we allow users to have made $k_u$ purchases before starting the system.
Suppose now that we are at time $t>0$, and let $I_u^t \subseteq [t-1]$ denote the set of times when $u$ previously made a purchase.
We denote the weight of a purchase made by $u$ at time $i<t$ as $h^{t,i}_u$. First, we assume these weights to be normalised, that is, for all $t>0$,
$$
\sum_{i=1}^{k_u}  h^{t,-i}_u + \sum_{i \in I_u^t} h^{t,i}_u = 1,
$$
For a given product $p$, let $J_{u,p}^t \subseteq \{-k_u, \ldots, -1 \} \cup I_u^t$ be the set of times when $u$ bought $p$. Then, given that $u$ is the recommender node at time $t$, product $p$ will be recommended by $u$ and purchased by the buyer with probability
$$
x^t_u(p) := \sum_{i \in J_{u,p}^t} h^{t,i}_u.
$$
This formalization captures a large spectrum of natural scenarios.
For example, if the recommendation consists in picking from $u$'s past purchases uniformly at random, this amounts to setting $h_u^{t,-i} = h_u^{t,i} = \nicefrac{1}{k_u+|I_{u}^t|}$ for all $i$.
Picking one at random among the $10$ most recent purchases corresponds instead to setting $h_u^{t,i} = \nicefrac{1}{10}$ for all $i$ such that $|I_{u}^i| \ge |I_{u}^t| - 10$ and $h_u^{t,i} = 0$ otherwise.
Note that the values of the weights $h_u^{t,i}$ depend on the outcome of the process, thus the $h_u^{t,i}$ are random variables themselves.

While our theorems will be given for the full model, it may help the intuition to consider the following simplified version, to which we will sometime refer in the remainder of the paper.
\begin{definition}\label{def:bs}
{\sc [The Basic Scenario]}
The special case where users buy at the same rate (i.e.\ $f_u = \nicefrac 1n, u \in {\cal U}$), follow recommendations with the same probability (i.e.\ $\alpha_u = \alpha, u \in {\cal U}$), pick the recommender with uniform probability (i.e.\ $w_{vu} = \nicefrac 1{\operatorname{indeg}(u)}, u \in {\cal U}$), and recommend items by picking a random product from their list with uniform probability (i.e.\ $h_u^{t,-i} = h_u^{t,i} = \nicefrac{1}{k_u+|I_{u}^t|}, u \in {\cal U}$), is called the {\em basic scenario}. 
\end{definition}

\noindent
{\bf Remarks.}
Our model only considers products that can be bought multiple times.
Food items are a natural example, but there are very many cultural products that fall into this category: going out to a specific club, attending sport events, or participating to concerts at specific venues are just a few examples.
The case of items for which multiple purchases make little sense, such as books or movies, is also of great interest but is left for future work. 
It is worth noting that our model can capture to a certain extent \emph{repeat consumption}, the empirical principle that people tend to re-purchase a product they recently bought~\cite{Anderson&2014,Benson&2016}.
More precisely, our model can mimic the one in~\cite{Anderson&2014}, where a user buys a product that is either new or drawn from a distribution over his/her purchasing history.
We can mimic this model by allowing self-loops in the graph: a user can then pick itself as the recommender node and, with properly defined weights, re-purchase a product he bought in the near past.

\subsection{Analysis of the model}\label{s:an}
To analyse the evolution of the process we focus, from now on, on a particular product $p^* \in \mathcal{P}$, and look at all the others as aggregated into a second product. 

Let us begin with some necessary notation.
Recall that $x_u^t(p)$ is the probability with which $p$ is recommended and bought given that $u$ is the recommender node at time $t$. To simplify notation, let $x_u^t(p^*)$ be denoted simply as $x_u^t$. Similarly, let $b_u$ be $u$'s personal preference for $p^*$, and let $\bm{b} = (b_1,\ldots,b_n)$.  We want to study the evolution of $\bm{x}^t := (x_1^t, \ldots, x_n^t)$ as $t$ grows.
We define  $\A = \operatorname{diag}(\alpha_1, \ldots, \alpha_n)$ and
denote by $\M$ the weighted transposed adjacency matrix of $G$, i.e.\ $\Mthin_{uv}=w_{vu}$ (which is $0$ if $vu \notin G$).
For any random vector $\bm{y} = (y_1, y_2, \ldots)$ we denote by $\E[\bm{y}]$ the vector $(\E[y_1], \E[y_2], \ldots)$ of the expectations of its components.

With the notation pinned down, we can now address our main question: does the market share of $p^*$ converge? 
The theorem below states that this happens if the weights $h_u^{t,i}$ satisfy the following condition.
\begin{itemize}\itemsep0pt
 \item[] \textit{The past fades away}, {\em i.e.} $\lim_{t\to \infty} \E[h_u^{t,i}] = 0$. 
\end{itemize}
This condition is very natural. It states that past purchases will eventually be forgotten. For instance it is satisfied when the recommended product is chosen uniformly at random or when only the last few purchases have non-zero weight.
We can now give our main result, which is proven in the appendix.
\begin{theorem}
\label{thm:convergence}
If the past fades away, then
\begin{align}
\lim_{t\to \infty} \E[\bm{x}^t] = \bm{x}^{\infty}
\end{align}
where $\bm{x}^{\infty} = (\I - \A\M)^{-1}(\I-\A)\bm{b}$.
Otherwise $\E[\bm{x}^t]$ might not converge, or converge but not to $\bm{x}^{\infty}$.
\end{theorem}
Let us discuss a few interesting consequences of this result.

First, notice that since the network and the values $\bm{\alpha} = (\alpha_1, \ldots, \alpha_n), \bm{f} = (f_1,\ldots,f_n)$, and $\bm{b}$, stay put and only $\bm{x}^t$ changes, the theorem gives a complete characterization of the market at the limit. 
The limit vector $\bm{x}^{\infty}$ has a natural and quite useful interpretation, given in the next corollary.
\begin{definition}\label{def:lms}
The fraction of times that a user $u$ has bought $p^*$ at time $t \in \mathbb{N}^+ \cup \{\infty\}$ is called the {\em local market share (LMS)} of $p^*$ w.r.t. $u$ at time $t$.
\end{definition}
\begin{corollary}\label{cor:1}
For every user $u$, $x^{\infty}_u$ is the local market share of $p^*$ w.r.t. $u$ at the limit.
\end{corollary}
The proof of this corollary in the general case is slightly technical but in the basic scenario (see Definition~\ref{def:bs}) follows immediately from Theorem~\ref{thm:convergence}.  In this scenario, by definition of the model, $x_u^t$ is the fraction of times that $u$ has bought $p^*$ up to time $t$, {\em i.e.} it is $p^*$'s LMS w.r.t. $u$, and Theorem~\ref{thm:convergence} states that this value converges to $x^{\infty}_u$. 

Consider now the matrix
\begin{equation}\label{eq:L}
 \bm{L} := (\I-\A\M)^{-1} (\I - \A)
\end{equation}
This matrix contains a lot of information about the steady state. Now, since $\bm{x}^{\infty} = \bm{L b}$ we have that $x^{\infty}_u = \sum_v L_{uv} b_v$. We can think of $x^{\infty}_u$ as consisting of a set of ``slices'', each of size $L_{uv} b_v$, $v \in {\cal U}$. Each slice is the contribution of $v$ to $x^{\infty}_u$ which is, as we have seen, $p^*$'s LMS w.r.t $u$ at the limit. We thus have a measure of how much each user can influence every other user.  

If $L_{vu}$ is how much $u$ influences $v$ then $\sum_v  L_{vu}$ quantifies $u$'s influence on the entire market. 
We can collect  these individual market influences in an \textit{influence vector}:
\begin{align}
\label{eqn:bgamma}
 \bm{\gamma}^T := \boldf^T \bL
\end{align}
where $\boldf = [f_1, \ldots, f_n]$.
The vector $\bgamma$ is stochastic (its entries sum up to one), and its element $\gamma_u$ is the relative weight of $u$'s personal distribution in the overall sales distribution at equilibrium. The larger the value of $\gamma_u$ the greater the influence of $u$ on the final market share of $p^*$.
We thus define the {\em maximum influence} of any user as,
\begin{align}
\label{eqn:gamma}
 \gamma(G) := \max_{u \in G} \gamma_u = \lVert \bgamma \rVert_{\infty}
\end{align}
For instance, in the basic scenario, when $G$ is a star $\gamma_u = \Theta(1)$ if $u$ is the centre and $\gamma_u = \Theta(\nicefrac 1n)$ otherwise, so that $\gamma(G)=\Theta(1)$, while when $G$ is a clique $\gamma(G)=\gamma_u = \nicefrac 1n$ for every $u$. 
 
Thanks to Theorem~\ref{thm:convergence} we can now easily answer our motivating question-- how much can the market be distorted? By combining $\bgamma$ and the vector $\bm{b}$ of personal preferences we can quantify how much the recommender system distorts (i.e.\ increases or decreases) the overall market share of $p^*$. 

We define the \textit{market distortion} as:
\begin{equation}\label{eq:D}
 \Delta := \frac{\boldf \cdot \bm{x}^{\infty}}{\boldf \cdot \bb} = \frac{\bgamma \cdot  \bb}{\boldf \cdot \bb},
\end{equation}
where $\cdot$ denotes the dot product. This is the ratio between the market share of $p^*$ at the limit and that without the recommender. 

\paragraph{Influence, PageRank and Shapley Values}
The influence vector $\bm{\gamma}$ has a natural interpretation in terms of personalized PageRank.
Consider the graph $G^T$ obtained by transposing $G$, the graph among the users, and suppose it has no dangling nodes (i.e.\ no node with outdegree zero).
Let $\bm{W}$ be the normalised adjacency matrix of $G^T$, i.e.\ $W_{uv} = \nicefrac{1}{\operatorname{outdeg}(u)}$ with the outdegree measured in $G^T$.
Finally, pick any $\alpha \in [0,1)$ and any $n$-entry stochastic vector $\bm{r}$.
The personalized PageRank vector $\bm{p}$ with respect to $\bm{r}$ with damping factor $\alpha$ on $G^T$ is defined as (see e.g.~\cite{LM04}):
\begin{equation}
 \bm{p}^T = \bm{r}^T (\I - \alpha \bm{W})^{-1} (1 - \alpha)
\end{equation}
Compare the above expression with that of $\bgamma^T$ given by Equation~\ref{eqn:bgamma} and the definition of $\bm{L}$:
\begin{equation}
 \bgamma^T = \boldf^T (\I-\A\M)^{-1} (\I - \A)
\end{equation}
The vector $\bgamma^T$ is the perfect analogue of $\bm{p}^T$.
Indeed, the two coincide by choosing $\bm{r} = \boldf$ if each node $u$ in $G$ has $\alpha_u = \alpha$ and has a nonempty set of incoming arcs all with the same weight.
Our notion of influence is thus a generalization of the standard personalized PageRank, allowing for personalized values of $\alpha$, graphs with dangling nodes, and arbitrary arc weights.
In fact, our model offers an intuitive interpretation similar to that of PageRank: the influence of nodes can be seen as the stationary distribution of a ``reverse'' random walk on $G$ which at each step, from the current node $u$, either moves to an incoming neighbour $v$ with probability $\alpha_u w_{vu}$ or, with probability $1-\alpha_u$, moves to a random node in $G$ chosen according to the distribution $\{f_v\}_{v \in G}$.
This walk can be seen as propagating the ``credit'' for the choice of a purchased product back to the incoming neighbours, the neighbours of the incoming neighbours and so on, recursively.
The purchasing rates play the intuitive role of the personalization vector: the more a user buys, the more the users who influence him are influential overall (and they are indeed reached more often by the random walk).

The influence vector $\bm{\gamma}$ also has an interpretation in terms of the classical notion of Shapley value, which is a mechanism to determine payoffs in coalition games. Shapley values have been proposed as a mechanism to determine fair compensation for user's participation to recommendation systems \cite{KPR2001}. Let $b_u$ denote, as usual, $u$'s personal preference for $p^*$.
Consider now the following game: each user can either participate to the market with $b_u$ or set its preference to $0$; the value assigned to the coalition is the total market share of $p^*$.
Then, the following holds.
\begin{fact}
The Shapley value of each user $u$ in the above game is exactly $\gamma_u b_u$.
\end{fact}
It is worth pointing out that, uncharacteristically, this Shapley value can be computed in polynomial-time.

\paragraph{Speed of convergence}
Theorem~\ref{thm:convergence} does not tell how much time it takes for the market to reach the steady-state; in fact, this depends on how fast the weights decay with $t$.
In the case of the simple uniform random recommendation, one can prove that the system converges to the steady state reasonably fast:
\begin{theorem}
\label{thm:rate_uniform}
 If the recommender node recommends an item picked at random from its purchasing history, then
\begin{align}
 \lVert \E[\x^t] -  \xlim \rVert = O\big(t^{\alpha - 1}\big)
\end{align}
where $\alpha$ is the maximum $\alpha_u$ over all users.
\end{theorem}
The proof is omitted from this extended abstract.

\paragraph{Computing $\bgamma$}
Clearly, computing $\bgamma$ efficiently is a crucial task to the aims of the present paper.
Doing this directly through Equations~\ref{eq:L} and \ref{eqn:bgamma} requires to explicitly invert the matrix $(\I - \A\M)$, which can be prohibitively expensive is the inverse is dense (which is the case if $G$ is well-connected).
We instead follow a standard alternative approach.
Rewrite $(\I - \A\M)$ as $\sum_{i=0}^{+\infty} (\A \M)^{i}$, plug it into Equation~\ref{eqn:bgamma} and obtain:
\begin{align}
\label{eqn:gamma_series}
 \bgamma^T = \boldf^T \sum_{i=0}^{+\infty} (\A \M)^{i} (\I-\A)
\end{align}
Now, if $G$ is sparse, so is $\M$ and thus all the terms of the summation in Equation~\ref{eqn:gamma_series}.
We can then approximate $\bgamma$ efficiently by truncating the series, which only requires to keep trace of a sparse matrix and vectors.
To be more precise, computing all entries of $\bgamma$ up to an additive error $O(\epsilon)$ requires time $O\big((|V|+|E|) \log(\frac{1}{\epsilon})/\log(\frac{1}{\alpha})\big)$ and space $O(|V|+|E|)$, where $\alpha = \max_{u\in G}{\alpha_u}$.
We use this technique throughout our experiments.

%% file: experiments.tex
\newcommand{\boldb}{\mathbf{b}}
\newcommand{\z}{\mathbf{z}}

\section{Experimental results}
\label{sec:experimental}
In this section we discuss a series of experiments where the underlying graph of our model is a real world social network. 

Theorem~\ref{thm:convergence} says that $\bm{x^t}$ converges to $\bm{x}^{\infty}$ in expectation. A first set of experiments shows that the expectation does captures what happens at the limit, {\em i.e.} the variance is small.
We then investigate the question of market distortion -- for a given product $p^*$, how much does its market share change at the limit?
A related question concerns so-called influencers. What is their influence and how much can they change the market?  We show that when the underlying network is a real world social network market distortion  is negligible and even ``celebrities'' cannot have  significant impact. Vice versa, there are other realistic situations, such as the presence of a super-node acting as a ``media tycoon'', where certain nodes can conquer a significant share of the market thanks to the recommendation mechanism. We also give a rule of thumb to identify top influencers on the basis of their 2-hop neighbourhood.
Finally, we explore recommenders not covered by our theoretical results, in which products are recommended with super-linear probabilities.
We find that the results vary widely from graph to graph, and depend also heavily on the starting share of products.

Let us now describe the experimental setup.

\subsection{Experimental setup}
\newcommand{\gplus}{\texttt{Google+}\xspace}
\newcommand{\tsnap}{\texttt{Twitter SNAP}\xspace}
\newcommand{\tlaw}{\texttt{Twitter LAW}\xspace}
\newcommand{\facebook}{\texttt{Facebook}\xspace}
\newcommand{\yelp}{\texttt{Yelp}\xspace}
\newcommand{\slashdot}{\texttt{Slashdot}\xspace}
We run our experiments on six following real-world graphs.
\begin{description}
	\item[\gplus] The ego-Gplus graph from SNAP Datasets~\cite{snap}. It contains $107,614$ nodes (users) and $13,673,453$ arcs representing a follower-to-followed relationship in the Google+ social network.
	\item[\tsnap] The ego-Twitter graph from SNAP Datasets. It contains $81,306$ nodes and $2,420,766$ follower-to-followed arcs crawled from Twitter.
	\item[\tlaw]
	The twitter-2010 graph by the Laboratory for Web Algorithmics, Univ.\ Milan~\cite{BoVWFI,BRSLLP}. It contains $41,652,230$ nodes and $1,468,365,182$ follower-to-followed arcs from Twitter. 
	\item[\slashdot] The soc-sign-Slashdot090221 graph from SNAP Datasets. It contains $82,144$ users and $549,202$ signed (i.e.\ ``friends'' or ``foes'') arcs between users. We kept only the $425,072$ ``friend'' arcs.
	\item[\yelp] A social graph extracted from friendships in the Yelp Dataset Challenge dataset.\footnote{\small\url{https://www.yelp.com/dataset_challenge/}} The graph is undirected and contains $365,759$ users and $1,288,031$ edges.\footnote{Actually, the original Yelp dataset contains $117$ spurious directed edges that are not reciprocated.}
	\item[\facebook] The graph from the ``List of links'' Facebook dataset by MPI-SWS~\cite{Viswanath&2009}.\footnote{\small\url{http://socialnetworks.mpi-sws.org/data-wosn2009.html}} The graph is undirected and contains $63,731$ users and $817,090$ edges.
\end{description}
When the graphs are directed we reverse the edges (except for \tlaw that was already oriented correctly). This is because in social graphs an edge $uv$ denotes the fact that $u$ ``trusts'' or ``follows'' $v$, while in our model the meaning is that $u$ influences $v$. 

Two kinds of experiments are carried out in the following: simulations and computations.
In the first, we simulate the purchasing process for $10,000 n$ steps and in the basic scenario (see Definition~\ref{def:bs}), unless otherwise specified.
The probability with which users follow a recommendation is set to $\alpha = 0.2$ for every user.
Note that this is a conservative estimate, since the available evidence suggests both the relative influence of friends' recommendations (see e.g.~\cite{Leskovec&2007,Ye&2012}) and the fraction of sales causally imputable to recommender systems (see~\cite{Sharma&2015}) to be significantly smaller.
We are thus intentionally putting the RS in favourable conditions to distort the market.

In computations, we use the adjacency matrix of the underlying graph and the other parameters of the model to compute explicitly the influence vector $\bgamma$ (and related quantities), using the techniques described at the end of Section~\ref{s:t}.

Our code was partly written in C++ using Eigen\footnote{\small\url{http://eigen.tuxfamily.org}} and the GNU Scientific Library\footnote{\small\url{http://www.gnu.org/software/gsl/}}, and partly in Java using the WebGraph framework\footnote{\small\url{http://webgraph.di.unimi.it/}}.
We ran it on a distributed-memory $50$-node cluster provided by CINECA\footnote{\small\url{http://www.cineca.it/en}}, where each node is equipped with two ten-core Intel Xeon CPUs and 128GiB of main memory.
The source code can be downloaded from {\url{https://github.com/Steven--/recommender}}.


\subsection{Convergence to the equilibrium}
The first experiment checks that the purchasing process unrolls as predicted by our theoretical results.
As in \S~\ref{s:an}, we focus on a single product $p^*$, with the remaining products coalesced into one. Figure~\ref{fig:simulations_linfinity_norm} shows that $\bm{x}^t$ quickly converges to $\bm{x}^{\infty}$ as $t$ grows (the plot follows a power-law). 
\begin{figure}[!h]
	\centering
	\includegraphics[scale=.6]{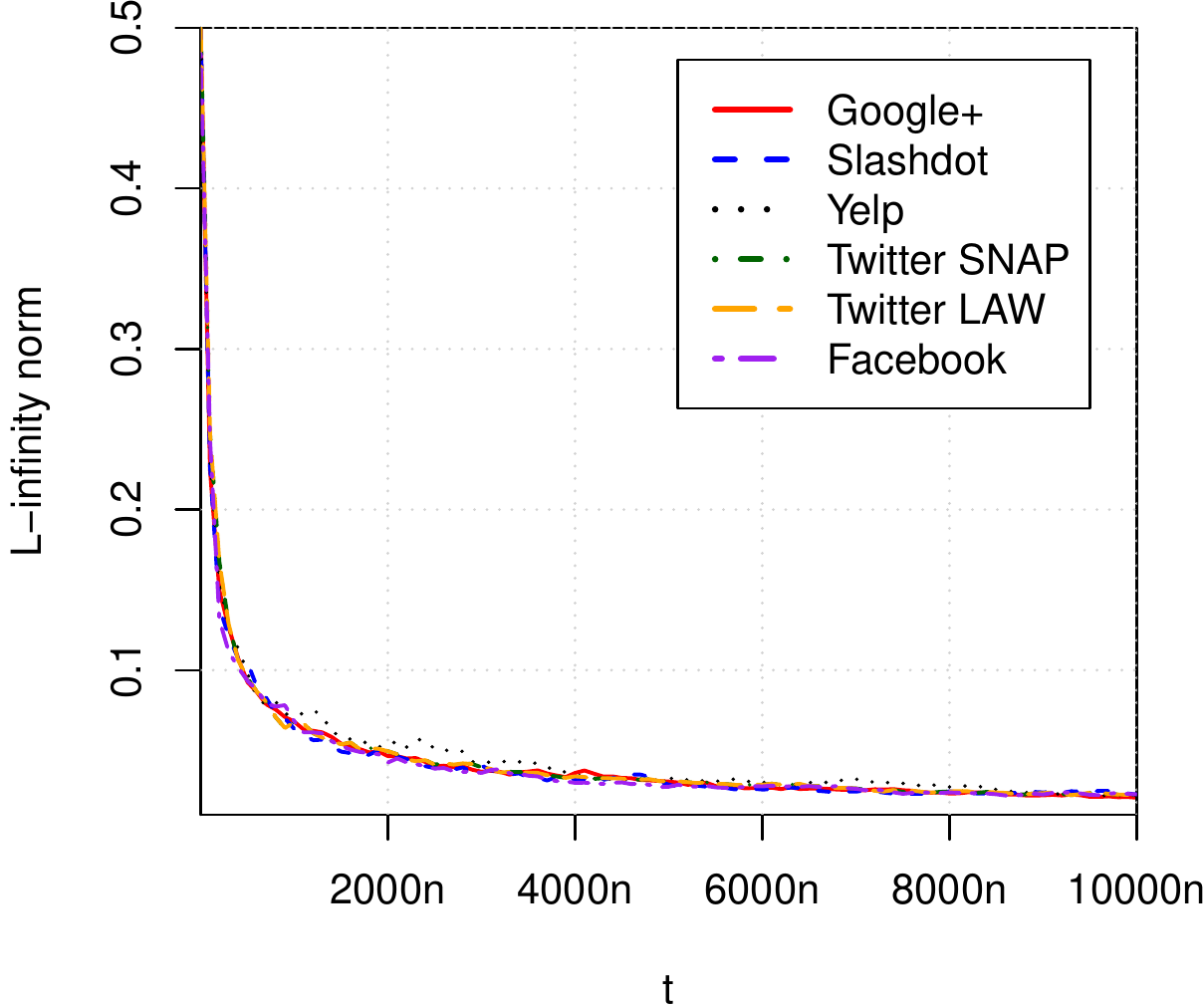}
	\caption{$\lVert \bm{x}^t - \bm{x}^\infty \rVert_\infty$ goes to $0$ as $t$ grows (long-memory scenario).}	\label{fig:simulations_linfinity_norm}
\end{figure}

\begin{figure}[t]
	\centering
	\includegraphics[scale=.6]{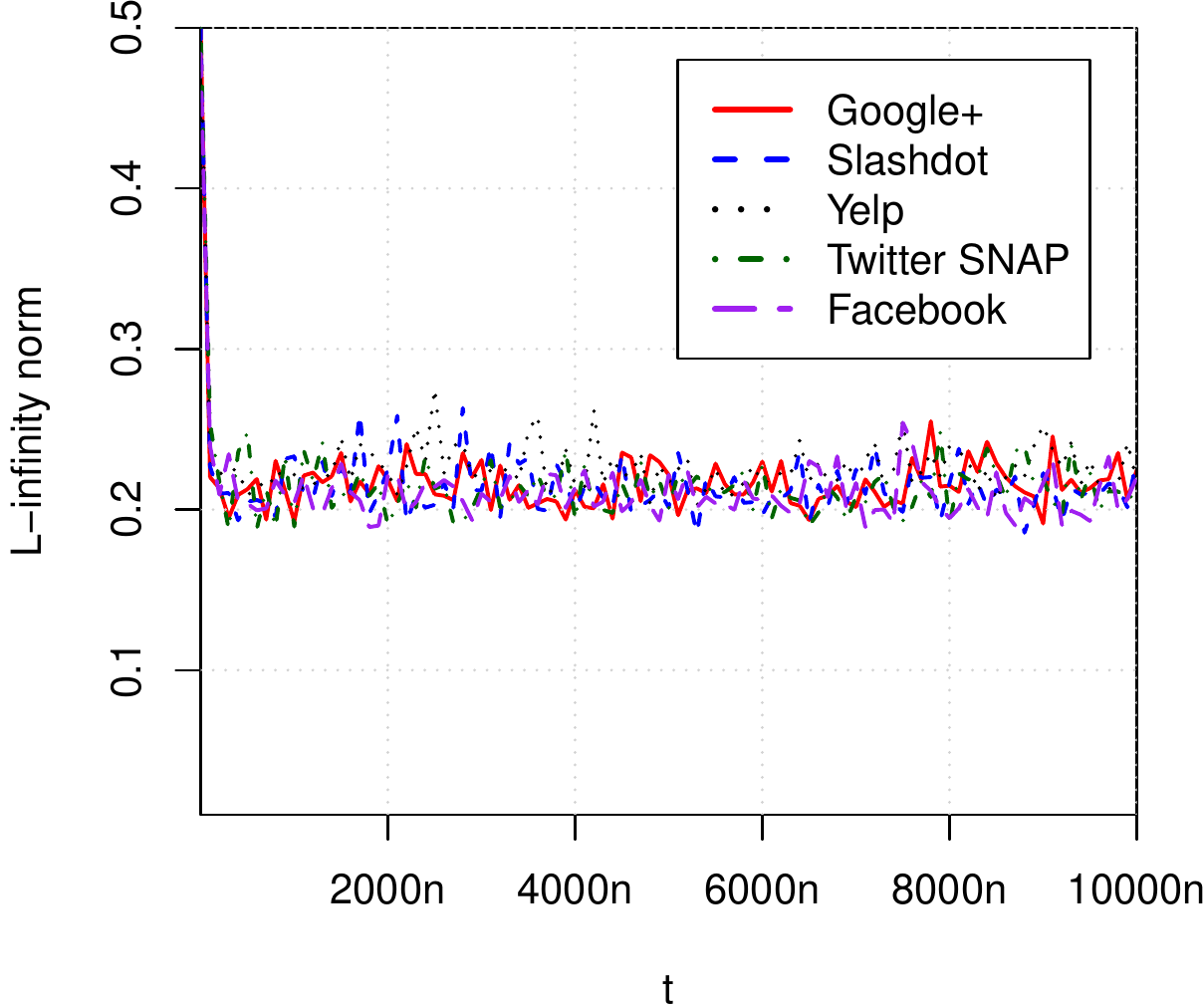}
	\caption{Plot of $\lVert \bm{x}^t - \bm{x}^\infty \rVert_\infty$ in the short-memory scenario.}
	\label{fig:simulations_linfinity_norm_h100}
\end{figure}

\begin{figure}[!t]
	\centering
	\includegraphics[scale=.6]{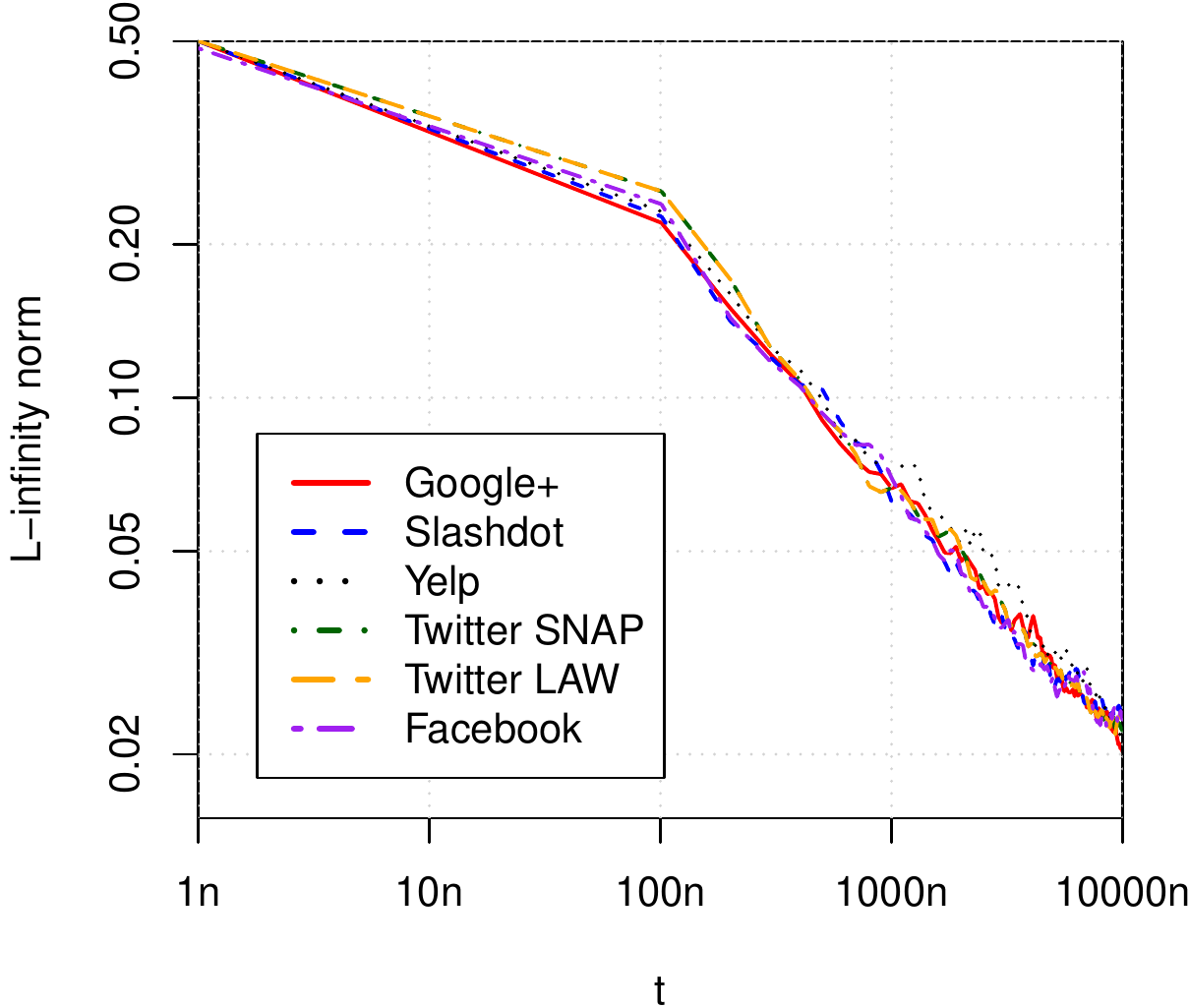}
	\caption{Log-log plot of $\lVert \bm{z}^t - \bm{x}^{\infty} \rVert_\infty$ in the short-memory scenario. The vector $\bm{z}^t$ keeps track, for each user, of the fraction of products $p^*$ bought since the beginning of the process.}
	\label{fig:simulations_linfinity_norm_log_log_h100}
\end{figure} 

Figure~\ref{fig:simulations_linfinity_norm_h100} shows what happens in the short-term setting. Initially, $\lVert \bm{x}_t - \bm{x} \rVert_\infty$  decreases rapidly, as expected, but then it starts oscillating around the value $0.22$. This behaviour might seem counter-intuitive at first, but it is a natural consequence of the short-memory of the users. As the users only consider the last $100$ purchases when making a recommendation, small differences in the number of times $p^*$ is purchased may considerably change the probability with which $p^*$ is recommended. This is a setting where $\E[\bm{x}^t]$ converges to $\bm{x}^{\infty}$ but $\bm{x}^t$ does not.
Note however the following.
If we consider the vector $\bm{z}^t$ containing the local market shares of 
$p^*$ w.r.t.\ every user $u$ (Definition~\ref{def:lms}) then, as predicted by Corollary~\ref{cor:1},
$\lVert \bm{z}^t- \bm{x}^{\infty} \rVert_\infty$ goes to 0 as $t$ increases
(see Figure~\ref{fig:simulations_linfinity_norm_log_log_h100}).
In conclusion, these experiments suggest that the LMS converges not only in expectation but in high probability.
We leave this as an open problem.

\subsection{Market distortion}
In this section we try to understand how much the market can be distorted by the recommendation system. To this end we consider the quantity
$\Delta$ defined by Equation~\ref{eq:D}
In the experiments $\bm{b}$ was instantiated with different types of distributions: uniform in $[0,1]$, exponential of mean $1/2$, a power-law of exponent $-0.01$ (i.e., $p(x) \propto x^{-0.01}$), and normal of mean $1/2$ and standard deviation $1/6$ (for distributions having support wider than $[0,1]$, the truncated version was used).

%
%
\begin{figure}[!b]
	\centering
	\includegraphics[width=0.95\columnwidth]{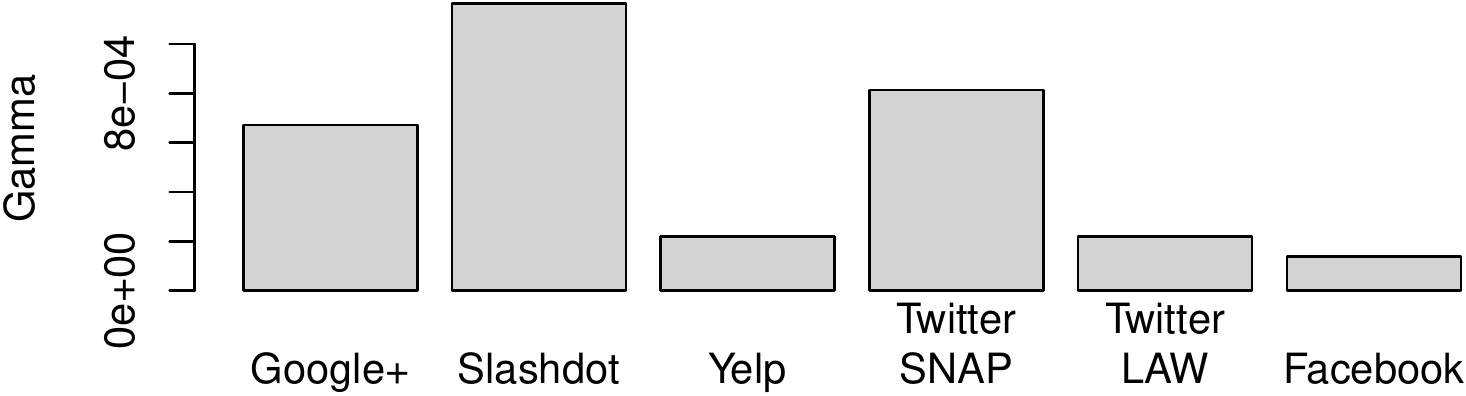}
	\includegraphics[width=0.9\columnwidth]{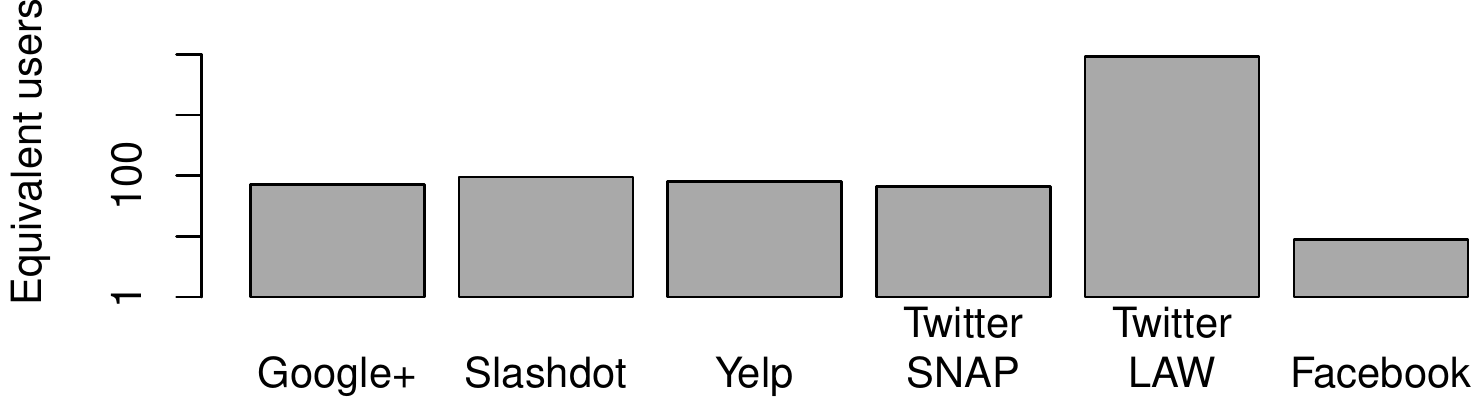}
	\caption{Fairness of social ties: values of $\gamma(G)$ (light grey) and $n \cdot \gamma(G)$ (dark grey)  on different social graphs.} 
	\label{fig:graphs_influence}
\end{figure} 
In all cases the value of $\Delta$ was always in the range $1 \pm 0.002$,  which means that market distortion was negligible for all networks considered.  This is an indication that social graphs prevent the recommender from distorting the market.

\begin{figure}[!t]
	\centering
	\includegraphics[width=0.9\columnwidth]{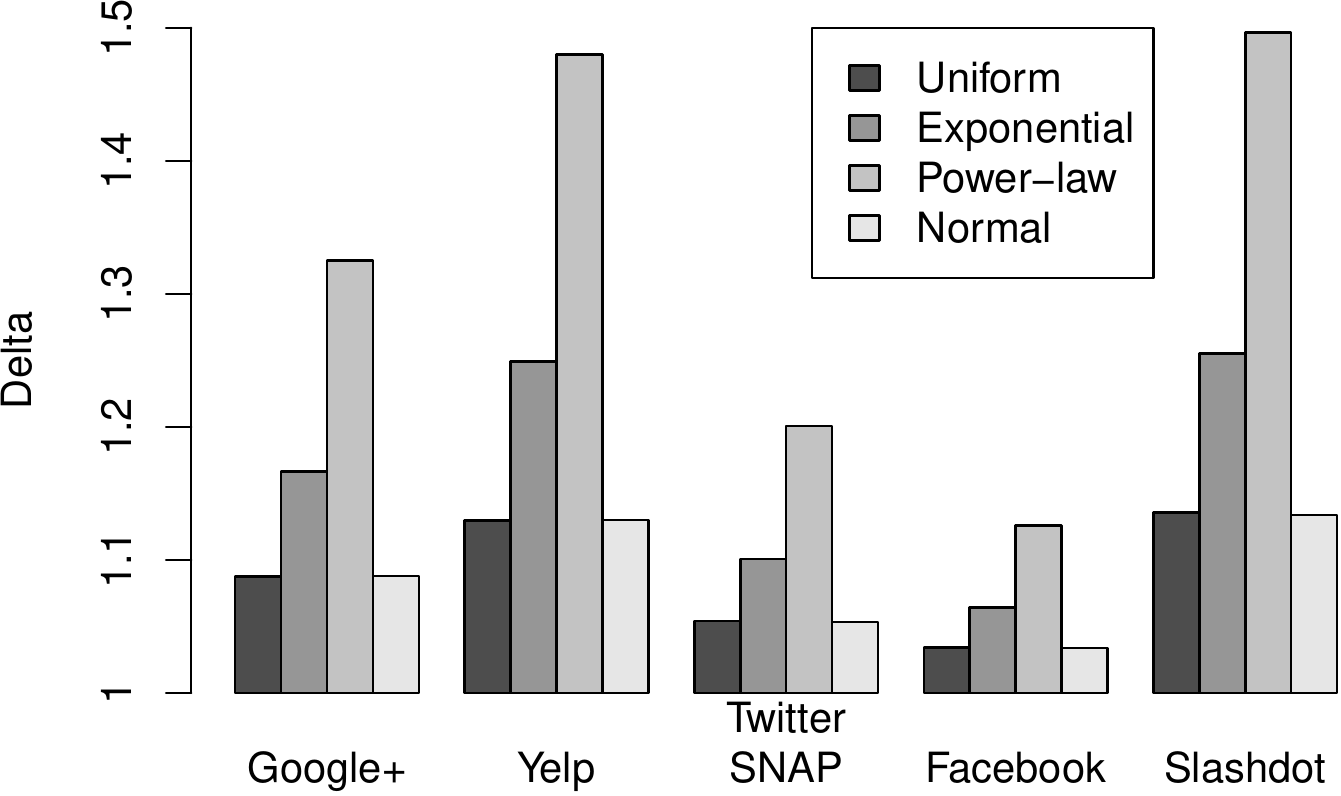}
	\caption{The effect of media tycoons: variation of the market share of $p^*$ (i.e., $\Delta$) for different graphs and starting distributions when a super-node is inserted.} 
	\label{fig:final_sales}
\end{figure} 

The next experiment determines how much (of the overall market share of $p^*$) can be ascribed to the most influential users. Recall from \S~\ref{s:an} that $\gamma(G) = \max_u \gamma_u$ is the maximum influence of any user.
While $\gamma(G)$ can be seen as a \emph{relative} measure of the importance of the most influential user, i.e., as the fraction of the network that a single user can influence,  $n \cdot \gamma(G)$ is an \emph{absolute} measure: the most influential user $u$ can be seen as controlling a group of $n \cdot \gamma(G)$ users that buy $p^*$ with probability $b_u$ at rate $1/n$. 
Figure~\ref{fig:graphs_influence} shows the values of $n \cdot \gamma(G)$ and of $\gamma(G)$ for the networks we tested. 
Even if there are users whose purchasing power corresponds to that of more than $80$ ``average users'', none of them can influence a significant portion of the network.

In order to understand if the fairness of these social graphs is an artefact of our model or one of their intrinsic properties, we investigated a modified version of our model not covered by Theorem~\ref{thm:convergence}. In the modified model, the probability with which $u$ recommends $p^*$ becomes super-linear and is proportional to $x^2_u$. 
Intuitively, this should boost the final market share of $p^*$ as soon as $p^*$ has a small advantage in the user's preference distributions. 
Therefore, we draw the personal preferences of the users so that the mean of $p^*$ is $k$ times larger than the mean of any other product, where the \emph{imbalance} $k$ ranges from $1$ to $5$.
We considered markets having both two and ten different products: the corresponding results can be seen in Figures~\ref{fig:nonlinear_beta2_prods2} and \ref{fig:nonlinear_beta2_prods10}. The purchasing process was simulated until the change in the market share of the products dropped below $0.001$, and the experiment was repeated $10$ times in order to compute the standard deviation of $\Delta$, which was always smaller than $0.007$. 
Even with this favourable setting we saw a market distortion effect on social networks that was always less than the one observed on a Clique. 
We conclude that, although the market share of the leading product $p^*$ can increase significantly when there are many competing products, social graphs do not promote, and often damps, this growth phenomena.  


\begin{figure}[!b]
	\centering
	\includegraphics[width=0.9\columnwidth]{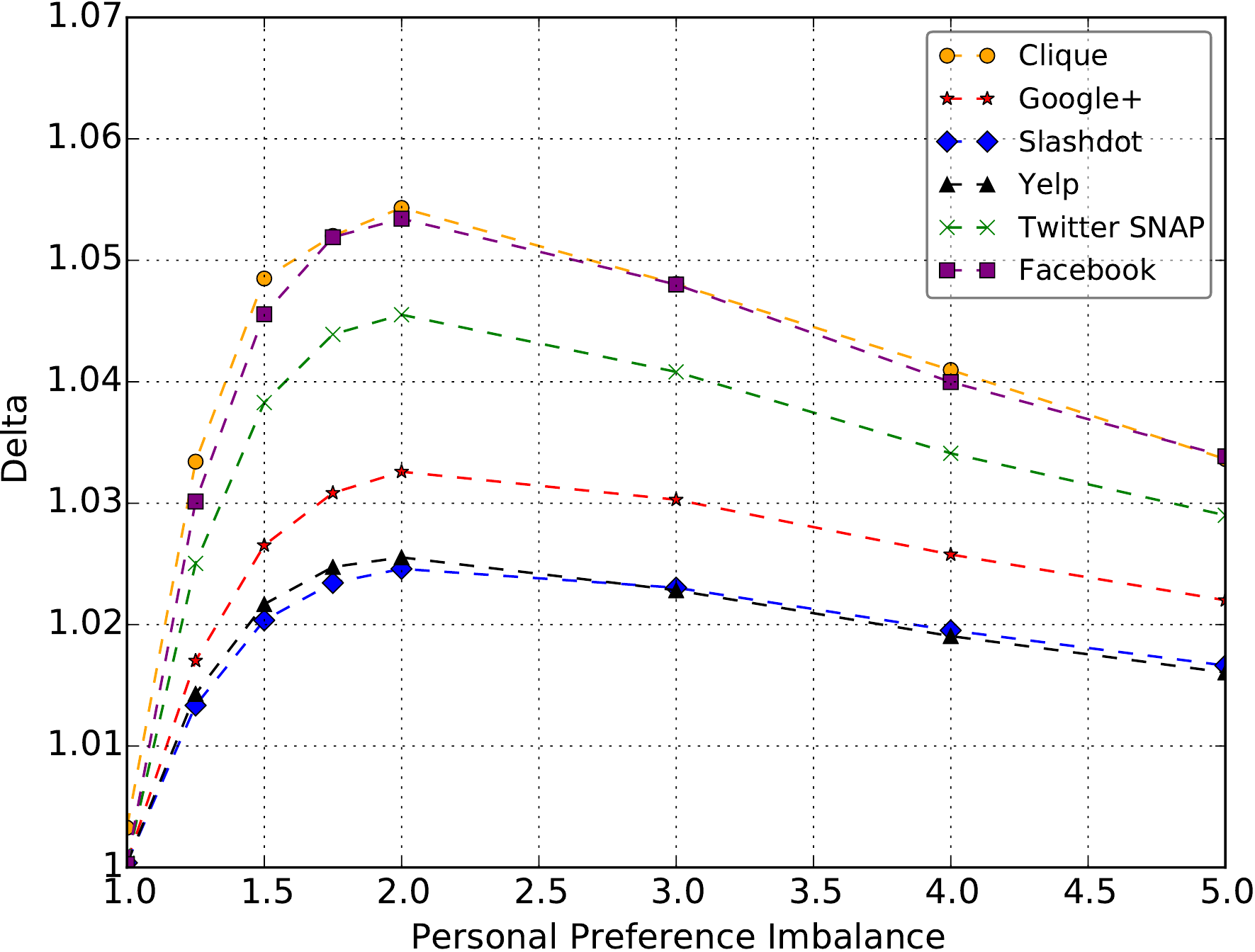}
	
	\caption{$\Delta$ for a market with $2$ products as a function of the initial imbalance for product $p^*$.} 
	\label{fig:nonlinear_beta2_prods2}
\end{figure} 

\begin{figure}[!t]
	\centering
	\includegraphics[width=0.9\columnwidth]{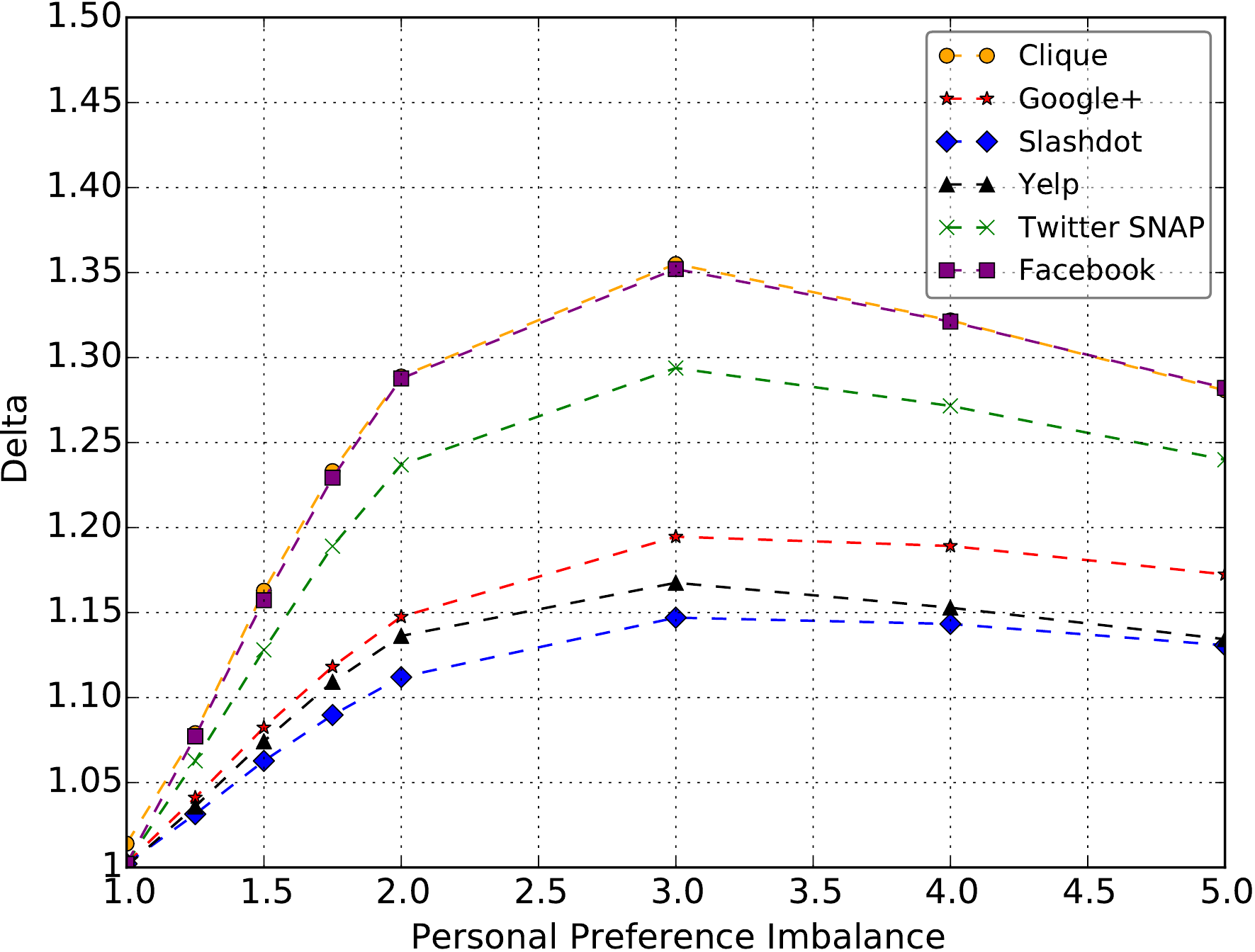}	
		\caption{$\Delta$ for a market with $10$ products as a function of the initial imbalance for product $p^*$.}  
	\label{fig:nonlinear_beta2_prods10}
\end{figure}

In the last experiment of this section we modify the topology of the networks by inserting a super-node that points to every user and always recommends $p^*$. This models the effects of mass media such as television, radio, mass advertising, etc. Perhaps
unsurprisingly, we observe that such ``media tycoons'' can influence a relevant portion of the market. Depending on the graph,  $\gamma(G)$ ranges from $3.4\%$ to $13.4\%$. Market distortion can be significant, with values of  $\Delta$ ranging from $105.3\%$ to $149.7\%$, as Figure~\ref{fig:final_sales} shows.

\subsection{Looking for the influencers}
In this section we look at the distribution of influence in social networks and give a rule of thumb to identify the most influential users.

In a clique every vertex has the same influence ($\gamma_u = \nicefrac 1n$ for each user $u$). We want to know how far real social networks are from this perfectly ``egalitarian'' situation.
Figure~\ref{fig:lorenz_curves} shows that real social networks are  quite fair 
in the sense that most of the users have an influence that is slightly smaller than $\nicefrac 1n$, while only  a small elite of users has a very large influence.
The graph where this ``oligarchy'' is most prominent is \tlaw, where $10\%$ of the overall influence on the market is due to $0.6\%$ of the users.  In all the other networks, at least $4\%$ of all the users are needed to reach the same combined influence.

\begin{figure}[!tb]
	\centering
	\includegraphics[width=0.9\columnwidth]{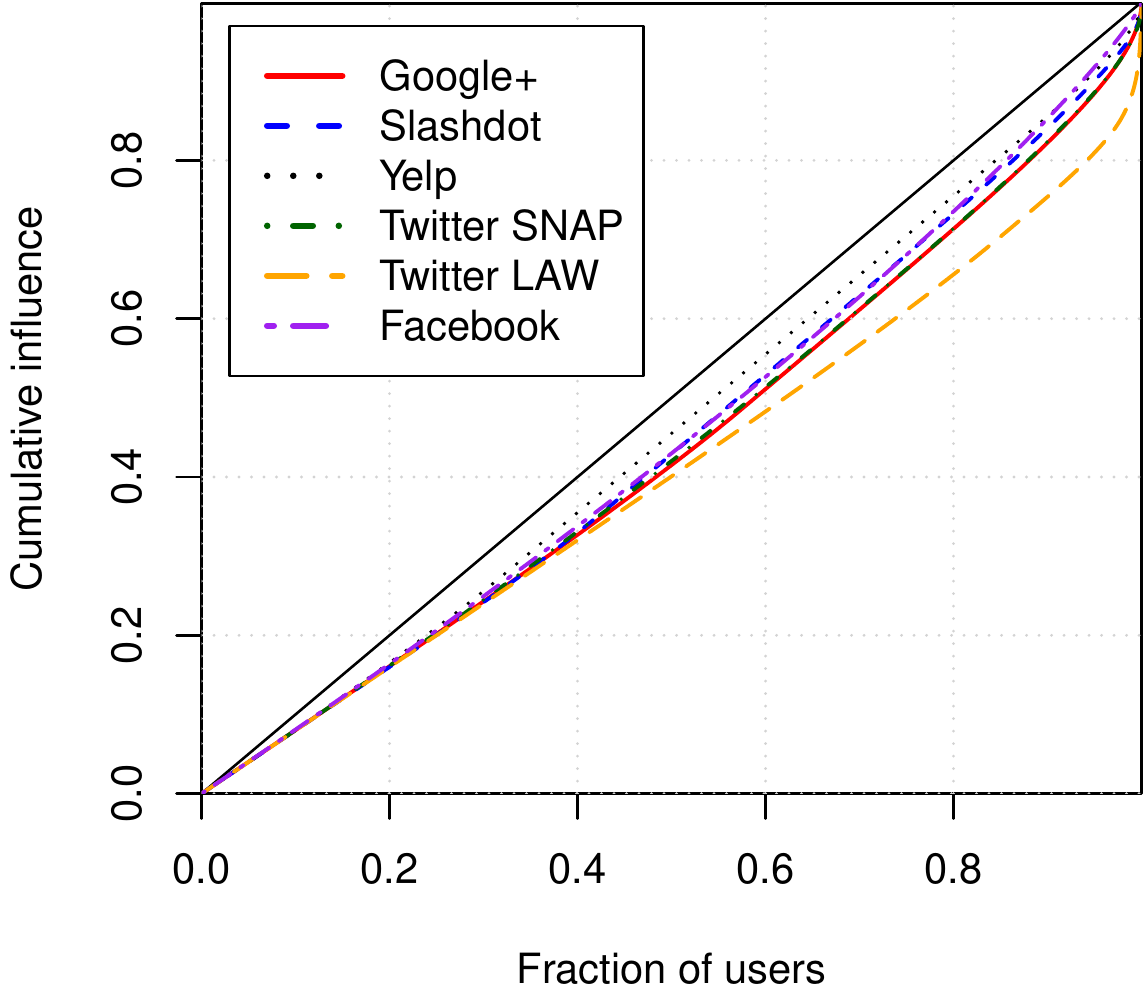}
	\caption{Lorenz curve of the vector $\gamma$ for the different graphs. The corresponding Gini indices are $0.128$, $0.109$, $0.076$, $0.125$, $0.180$, $0.101$, in order.}
	\label{fig:lorenz_curves}
\end{figure}



Next, we try to give a simple rule of thumb to identify the most influential users in a network. The bottom of
Figure~\ref{fig:graphs_gamma_degree} shows the first one hundred nodes ranked by out-degree from left to right on the $x$-axis. The $y$-axis reports their $\gamma$-rank (the shorter the bar the higher the rank). The figure shows that if a node has high out-degree than it will be influential. 
The top of Figure~\ref{fig:graphs_gamma_degree} shows that the converse is not true. Among the one hundred most influential users there are nodes with moderate to low out-degree. Who are these nodes? 
\begin{figure}[!h]
	\centering
	\includegraphics[scale=.5]{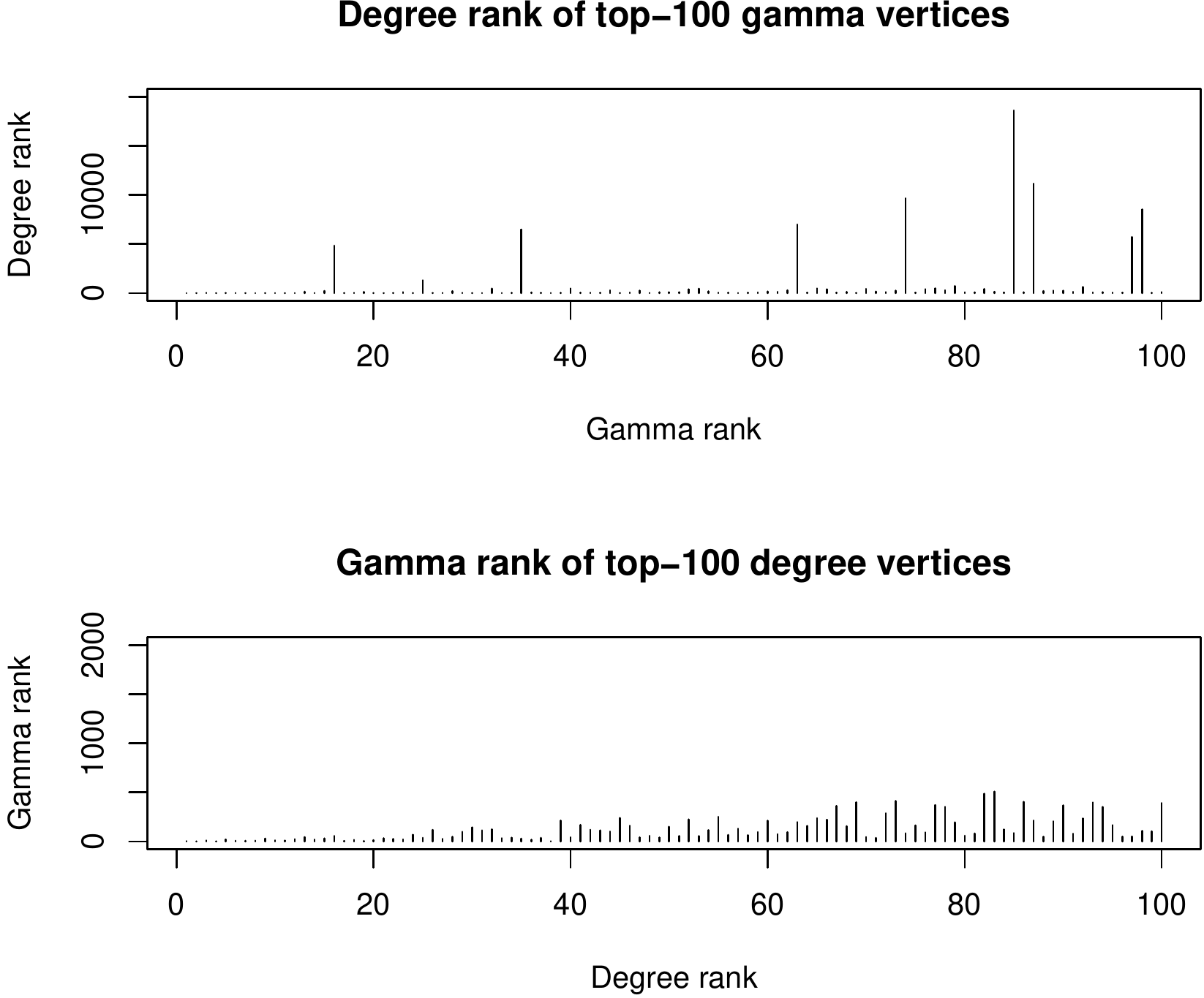}
	\caption{Top: Outdegree rank for the top-100 ranked users w.r.t.\ $\gamma$. Bottom: $\gamma$ rank for the top-100 ranked users w.r.t.\ their outdegree. The data refer to \gplus.}
	\label{fig:graphs_gamma_degree}
\end{figure}
To illustrate, consider the case of the user whose $\gamma$-rank in \gplus\ is $88$. This user lies within the top $1\%$ of the influential users, but his/her $243$ followers place him/her in position $11,155$ of the out-degree ranking, outside of the top $10\%$.
This user however has $48$ followers each of which follows $3$ users or less and, among these, $34$ happen to follow only him. In other words, this user is highly ``trusted'' in a small community. When this happens, the node will have a high $\gamma$-rank.

Finally, consider a user $u$ who is trusted by a user $v$ having high degree. We know already that $v$ will be influential. Suppose now that $v$ follows a small number of people that is, $u$ is among the few ``confidants'' of an influential user. Then, $u$ too will be influential. 

To summarize, we can characterize influential users as users who have at least one of the following properties:
\begin{itemize}
	\item A large number of followers.
	
	\item A non-negligible group of followers which, in turn, follow few other users.
	
	\item Are among the few friends of an influential user.
\end{itemize}
This rule of thumb gives a heuristic characterization of all influential users.

%% file: appendix.tex
\label{sec:appendix}

\subsubsection*{Proof of Theorem~\ref{thm:convergence}.}
\label{apx:thm_convergence}
We make use of the following lemma:
\begin{lemma}
\label{lem:invertible}
The matrix $(\I - \A\M)$ is nonsingular, and the matrix $(\I - \A\M)^{-1}(\I-\A)$ is row stochastic (i.e.\ is real, non-negative, and each row sums to $1$).
\end{lemma}
\begin{proof}
Since $\alpha_u < 1$ for all $u$ and $\M$ is row-stochastic, then $(\I - \A\M)$ is strictly diagonally dominant; hence it is invertible by the Levy-Desplanques theorem.
Now $(\I - \A\M)^{-1} = \sum_{i \geq 0} (\A\M)^i$ is real and nonnegative, as well as $(\I-\A)$, hence so is their product $(\I - \A\M)^{-1}(\I-\A)$.
Finally, since $\M\cdot\1 = \1$, then also $(\I-\A)\cdot\1 = (\I-\A\M)\cdot\1$, and thus $(\I - \A\M)^{-1}(\I-\A) \cdot \1 = (\I - \A\M)^{-1}(\I-\A\M) \cdot \1 = \1$.
\end{proof}
Lemma~\ref{lem:invertible} implies that the vector $\xlim$ exists and has components in $[0,1]$.
Let us now turn to prove the main statement.

\textsc{``if'' side.}
This part is divided in three steps: writing a recurrence equation for $\E[x_u^t]$, converting it into a matricial recurrence equation for $\E[\x^t] - \xlim$, and showing that $|\E[\x^t] - \xlim|$ tends to the null vector $\0$.

\textbf{A recurrence equation for $\E[x_u^t]$.}
The proof makes use of two additional random variables:
$\psi_u^t$, which indicates the event that $u$ buys at step $t$, and
$\chi_u^t$, which indicates the event that $u$ buys $p^*$ at step $t$.
Then $x_u^t$ can be written as:
\begin{equation}
\label{eqn:x_definition}
 x_u^t = \sum_{i=1}^{k_u} h_u^{t,-i} \chi_u^{-i} \; + \; \sum_{i=1}^{t-1} h_u^{t,i} \chi_u^{i}
\end{equation}
where the $\chi_u^{-i}$ are actually constants since the initial history is fixed.
Gather the first summation in a single random variable $h^{t,0}_u$, and take the expectation of both sides:
\begin{align}
\E[x_u^t]
&= \E[h_u^{t,0}] \; + \; \sum_{i=1}^{t-1} \E[ h_u^{t,i} \chi_u^{i} ]
\label{eqn:x_exp_recur_weig}
\end{align}
Note that $\E[\,\cdot\, z] = \E[\,\cdot\, |z=1]\E[z]$ if $z$ is a binary random variable.
Also, the weights do not depend on the specific products bought, so $\E[h_u^{t,i} | \chi_u^{i} = 1] = \E[h_u^{t,i} | \psi_u^{i} = 1]$.
Using these facts, one can easily rewrite $\E[h_u^{t,i} \chi_u^{i}]$ as $\E[h_u^{t,i} \psi_u^{i}] \E[\chi_u^{i} | \psi_u^{i} = 1]$.
We can then restate Equation~\ref{eqn:x_exp_recur_weig} in the following way:
\begin{align}
\E[x_u^t] &= \E[h_u^{t,0}] \; + \; \sum_{i=1}^{t-1} \E[ h_u^{t,i} \psi_u^{i}] \E[\chi_u^{i} | \psi_u^{i} = 1]
\label{eqn:x_exp_recur_weig_3}
\end{align}

Let us now focus on $\E[\chi_u^{i} | \psi_u^{i} = 1]$, the probability that $u$ buys $p^*$ at step $i$ given that he buys.
By definition of the purchasing process, we have:
\begin{align}
\E[\chi_u^{i}|\psi_u^{i}\!=\!1]
&= (1-\alpha_u) b_u + \alpha_u \!\!\!\!\! \sum_{(v,u) \in G}\!\!\!\!\! w_{vu}\, \E[x_v^i|\psi_u^{i} = 1]
\label{eqn:E_chi_psi}
\end{align}
However, we can replace $\E[x_v^i|\psi_u^{i} = 1]$ by $\E[x_v^i]$ as $x_v^i$ only depends on $v$'s past purchases.
Plugging the resulting expression for $\E[\chi_u^{i}|\psi_u^{i} = 1]$ into Equation~\ref{eqn:x_exp_recur_weig_3}, we obtain the following recurrence relation:
\begin{equation}
 \label{eqn:x_recur}
 \E[x_u^t] = \E[h_u^{t,0}] \; + \; \sum_{i=1}^{t-1} \E[h_u^{t,i} \psi_u^{i}] ( (1-\alpha_u) b_u + \alpha_u \!\!\!\! \sum_{(v,u) \in G}\!\!\!\! w_{vu}\, \E[x_v^i])
\end{equation}

\textbf{A recurrence equation for $\E[\x^t] - \xlim$.}
We next turn Equation~\ref{eqn:x_recur} into an expression for $\E[\x^t]$.
For all $t > 0$ let $\bh_{t,0}$ be the vector $[\E[h_1^{t,0}], \ldots, \E[h_n^{t,0}]]$;
similarly, for all $t>i>0$ let $\bH_{t,i}$ be the matrix $\operatorname{diag}(\E[h_1^{t,i}\psi_1^{i}], \ldots, \E[h_n^{t,i}\psi_n^{i}])$.
Finally, recall the definitions of $\A, \M, \bb, \x^i$.
The following matricial counterpart to Equation~\ref{eqn:x_recur} holds:
\begin{align}
\label{eqn:x_recur_mat}
 \E[\x^t] = \bh_{t,0} + \sum_{i=1}^{t-1} \bH_{t,i} ( (\I-\A)\bb + \A \M \E[\x^i] )
\end{align}
We now subtract $\xlim$ from both sides -- but putting the right-hand side in a convenient form.
Indeed, by definition of $\bh_{t,0}$ and $\bH_{t,i}$ and by the fact that $\sum_{i=1}^{k_u} h_u^{t,-i} + \sum_{i=1}^{t-1} h_u^{t,i} \psi_u^{i} = 1$, if we define $\bh_{t,0}'$ as the vector whose $u$-th component is $\E[h_u^{t,0}] - \E[\sum_{i=1}^{k_u} h_u^{t,-i}] x_u^{\infty}$, one can write:
\begin{align}
\label{eqn:barx_recur_mat}
 \E[\x^t] - \xlim = \bh_{t,0}' + \sum_{i=1}^{t-1} \bH_{t,i} ( (\I-\A)\bb + \A \M \E[\x^i] - \xlim)
\end{align}
But by definition of $\xlim$ we have $(\I-\A)\bb = (\I - \A\M)\xlim$; thus $(\I-\A)\bb - \xlim = -\A\M\xlim$ and Equation~\ref{eqn:barx_recur_mat} becomes:
\begin{align}
\label{eqn:x_recur_mat_3}
 \E[\x^t] -  \xlim &= \bh_{t,0}' + \sum_{i=1}^{t-1} \bH_{t,i} \A\M(\E[\x^i] - \xlim)
\end{align}

\textbf{$|\E[\x^t] -  \xlim|$ goes to zero.}
We finally show that $|\E[\x^t] -  \xlim|$, the vector of the absolute values of $\E[\x^t] -  \xlim$, vanishes with $t$.
The intuition is that $\bh_{t,0}'$ and some of the first terms of the summation in Equation~\ref{eqn:x_recur_mat_3} go to zero with $t$ by hypothesis, and the rest is damped by the norm of $\A$ which is strictly smaller than $1$; this gives a bound on $|\E[\x^t] -  \xlim|$ which can be plugged in again, proving the thesis by induction.

Hereafter, inequalities are meant component-wise.
First, since $\bH_{t,i}$, $\A$, $\M$ are non-negative, Equation~\ref{eqn:x_recur_mat_3} implies:
\begin{align}
\label{eqn:x_recur_norm}
 |\E[\x^t] -  \xlim| &\le |\bh_{t,0}'| + \sum_{i=1}^{t-1} \bH_{t,i} \A \M | \E[\x^i] - \xlim |
\end{align}
Suppose now $|\E[\x^t] -  \xlim| \le c\1$ for all $t > \tau$, for some constants $\tau, c \geq 0$; this trivially holds for $\tau = 0$ and $c = 1$.
One can check that, by hypothesis of forgetfulness, each element of $\bh_{t,0}'$ goes to zero with $t$, and the same does each element of $\bH_{t,i}$ for any fixed $i$.
Thus for any given $\epsilon > 0$ there exists a $t' \geq \tau$ such that for all $t > t'$:
\begin{align}
 |\bh_{t,0}'| + \sum_{i=1}^{\tau} \bH_{t,i} \A \M | \E[\x^i] - \xlim| \le \epsilon
\end{align}
and then Equation~\ref{eqn:x_recur_norm} gives, for all $t > t'$:
\begin{align}
\label{eqn:x_recur_norm_3}
 |\E[\x^t] - \xlim| &\le \epsilon \1 + \sum_{i=\tau+1}^{t-1} \bH_{t,i} \A \M | \E[\x^i] - \xlim |\\
&\le \epsilon \1 + \sum_{i=\tau+1}^{t-1} \bH_{t,i} \A\M \cdot c\1
\label{eqn:x_recur_norm_4}
\end{align}
Each row of $\A\M$ sums to at most $\alpha$ where $\alpha = \max_{v\in G}{\alpha_v}$; while, since $\sum_{i=1}^{t-1} h_u^{t,i} \psi_u^{i} \leq 1$, each row of $(\sum_{i=\tau+1}^{t-1} \bH_{t,i})$ sums to at most $1$.
Each row of $\sum_{i=\tau+1}^{t-1} \bH_{t,i} \A\M$ thus sums to at most $\alpha$.
And then for all $t > t'$:
\begin{align}
\label{eqn:x_norm_bound}
 |\E[\x^t] - \xlim| &\le \epsilon \1 + \alpha c \1
\end{align}
which for e.g.\ $\epsilon = \frac{1-\alpha}{2} c$ gives $|\E[\x^t] - \xlim| \le \frac{1+\alpha}{2} c \1$.
But we can then repeat the whole argument multiple times, each time replacing $\tau$ by $t'$ and $c$ by $\frac{1+\alpha}{2}c$ and tightening the bound on $|\E[\x^t] - \xlim|$ by a factor $\frac{1+\alpha}{2}$ for all sufficiently large $t$.
Since $\alpha < 1$ and thus $\frac{1+\alpha}{2} < 1$, this implies
\begin{align}
\label{eqn:limit}
 \lim_{t \rightarrow +\infty} |\E[\x^{t}] - \xlim| = \0
\end{align}
which concludes this part of the proof.

\textsc{``only if'' side}.
We show a simple counterexample.
Pick a $G$ where all $u$ have incoming arcs and let all $\alpha_u > 0$.
Also, let all $u$ buy the same product $a$ at time $-1$ and set $h_u^{t,-1} = 1$ for every $t$.
All nodes thus always recommend $a$, so if $\lim_{t \to \infty} \E[\x^t]$ exists it must depend on whether or not $p^* = a$; thus it cannot be $\xlim$.
Moreover, one can make the system never converge by putting two different products $a$ and $b$ in the initial purchasing history of all users and making them periodically recommend only $a$ or only $b$.